\documentclass{article}
\usepackage{amssymb,amsmath}
\usepackage{graphicx,enumerate}

\title{Possibilistic investment models with background risk}

\author{Irina Georgescu \\ \footnotesize Academy of Economic Studies\\ \footnotesize Department of Economic Cybernetics\\ \footnotesize Pia$\c{t}$a Romana No 6  R 70167, Oficiul Postal 22, Bucharest, Romania \\
 \footnotesize Email: irina.georgescu@csie.ase.ro}

\date{}

\begin{document}
\maketitle

\begin{abstract}
In the study of investment problem, aside from the investment risk the background risk appears. Both the investment risk and the background risk are probabilistically described by random variables. This paper starts from the hypothesis that the two types of risk can be represented both probabilistically (by random variables) and possibilistically (by fuzzy numbers). We will study three models in which the investment risk and the background risk can be: fuzzy numbers, a random variable -- a fuzzy number and a fuzzy number -- a random variable. A portfolio problem is formulated for each model and an approximate calculation formula of the optimal solution is proved.
\end{abstract}


\newtheorem{definitie}{Definition}[section]
\newtheorem{propozitie}[definitie]{Proposition}
\newtheorem{remarca}[definitie]{Remark}
\newtheorem{exemplu}[definitie]{Example}
\newtheorem{intrebare}[definitie]{Open question}
\newtheorem{lema}[definitie]{Lemma}
\newtheorem{teorema}[definitie]{Theorem}
\newtheorem{corolar}[definitie]{Corollary}

\newenvironment{proof}{\noindent\textbf{Proof.}}{\hfill\rule{2mm}{2mm}\vspace*{5mm}}

\section{Introduction}

The effect of background risk on investors'decisions is a topic
which appears frequently in the risk management literature (see
\cite{eeckhoudt2}, \cite{guido} for several references). The models
studied in \cite{doherty}, \cite{eeckhoudt1}, \cite{eeckhoudt2},
\cite{guido}, \cite{laffont} study the investment risk and the
background risk, their influence on the choice of optimal
portfolios, the calculation of optimal solutions, etc. . A
fundamental hypothesis of these papers is that both the investment
risk and the background risk are random variables.

On the other hand, in the last years risk management models built in
the framework of Zadeh's possibility theory \cite{zadeh1} appeared.
Possibility distributions replace random variables, and usual
probabilistic indicators (expected value, variance, covariance) are
replaced by adequate possibilistic indicators \cite{appadoo},
\cite{carlsson2}, \cite{carlsson3}, \cite{dubois2}, \cite{fuller1}.
An important class of possibility distributions is represented by
fuzzy numbers \cite{dubois1}, \cite{dubois2} and most of these
models are based on them.

This paper proposes an approach of background risk models in a possibilistic context. We will notice four ways of connecting the investment risk and the background risk in the building of a model.

\begin{table}[ht]
\centering
\begin{tabular}{c c c}
\hline \hline
   & Investment risk & Background risk \\[0.5ex]
\hline
1 & probabilistic & probabilistic \\
2 & possibilistic  & possibilistic \\
3 & possibilistic & probabilistic \\
4 & probabilistic & possibilistic \\[1ex]
\hline
\end{tabular}
\label{table: nonlin}
\end{table}

Case $1$ is treated in \cite{eeckhoudt2}, p. 68. The other three cases will be studied in this paper. For each model an optimization problem in terms of possibilistic expected utility \cite{georgescu3} or mixed expected utility \cite{georgescu4} will be formulated and an appoximate calculation formula of the solution will be proved.

The paper is organized as follows.

Section $2$ presents fuzzy numbers, their operations and indicators (\cite{carlsson2}, \cite{dubois1}, \cite{dubois2}, \cite{georgescu2}). In Section $3$ multidimensional possibilistic expected utility, mixed expected utility and some of their properties (\cite{georgescu3}, \cite{georgescu4}) are recalled. These two notions will be used for defining the objective functions of the models in Sections $5-7$.

Section $4$ contains a two--asset possibilistic model without the background risk. It is the possibilistic analogue of the probabilistic model of
\cite{eeckhoudt2} and it represents a starting point for the topics developed in the next sections. Investing an initial wealth between a risk--free asset (bonds) and a risky asset (stocks), an agent should find out a portfolio bringing him a maximal gain. The model is characterized by the fact that the return of stocks is a fuzzy number and not a random variable as in \cite{eeckhoudt2}. A portfolio problem in possibilistic terms is formulated, its optimal solution and its approximate calculation are studied.

In Section $5$ this investment model is enriched with possibilistic background risk. Here both the investment risk and the background risk will be fuzzy numbers. The portfolio value will be expressed with respect to these two fuzzy numbers and the objective function of the model will be a possibilistic expected utility in the sense of \cite{georgescu3}. An approximate value $\alpha^{\ast \ast}$ of the optimal solution will be expressed w.r.t.
the approximate value $\alpha^\ast$ of the optimal solution of the problem of Section $4$ and some possibilistic indicators of fuzzy numbers (expected value, variance, covariance).

The investment models in the next sections combine probability theory and possibility theory: in Section 6 the investment risk is a fuzzy number and the background risk is a random variable, while in Section 7 the investment risk becomes a random variable and the background risk becomes a fuzzy number. The objective functions of the two models are built as mixed expected utilities \cite{georgescu4} and the approximate solutions of the optimization problems are expressed by combinations of possibilistic and probabilistic indicators.

In all the models of this paper the way the optimal solution changes with the investor's risk aversion is also studied.

\section{Possibilistic indicators of fuzzy numbers}

In this section we recall from \cite{dubois1}, \cite{dubois2}, \cite{carlsson2} the definition of fuzzy numbers, their operations and some of their properties. Also we will present the main indicators of fuzzy numbers (expected value, variance and covariance) (see \cite{carlsson1}, \cite{carlsson3}, \cite{dubois3}, \cite{majlender}, \cite{thavaneswaran1}, \cite{zhang2}).

Let $X$ be a set of states. A {\em fuzzy subset} of $X$ is a function $A: X \rightarrow [0, 1]$. A fuzzy set $A$ is {\em normal} if there exists $x \in X$ such that $A(x)=1$.  The {\em support} of  a fuzzy set $A$ is $supp(A)=\{x \in X| A(x)>0\}$.

We consider $X={\bf R}$. For $\gamma \in [0, 1]$, the $\gamma$--level set $[A]^\gamma$ of a fuzzy subset $A$ of ${\bf R}$ is defined by

$[A]^\gamma=\left\{\begin{array}{rcl}
\{x \in {\bf{R}}| A(x) \geq  \gamma\}& \mbox{if}& \gamma > 0 \\
cl (supp(A)) & \mbox{if}& \gamma =0
\end{array}\right.$

cl(supp(A)) is the topological closure of the set $supp(A) \subseteq {\bf R}$.

The fuzzy set $A$ is called {\em fuzzy convex} if $[A]^\gamma$ is a convex subset in ${\bf R}$ for any $\gamma \in [0, 1]$.

\begin{definitie}
A fuzzy subset $A$ of ${\bf R}$ is called {\em fuzzy number} if $A$ is normal, fuzzyconvex, continuous and with bounded support.
\end{definitie}

If $A$ is a  fuzzy number, $[A]^\gamma=[a_1(\gamma), a_2(\gamma)]$ for all $\gamma \in [0, 1]$ and $[a_1(0), a_2(0)]$ is the support of $A$. A {\em fuzzy point} is a fuzzy number $A$ with the support having one element.

Let $A, B$ be two fuzzy numbers and $\lambda \in {\bf R}$. By applying Zadeh's extension principle \cite{zadeh2} we define the fuzzy numbers  $A+ B$   and $\lambda A$ by

$(A+B)(x)=\displaystyle \sup_{y+z=x} \min (A(y), B(z))$

$(\lambda A)(x)=\displaystyle \sup_{\lambda y=x} A(y)$

In this way the operations with real numbers extend to operations with fuzzy numbers. Most of the properties of real numbers are preserved for real numbers.

If $[A]^\gamma=[a_1(\gamma), a_2(\gamma)]$, $[B]^\gamma=[b_1(\gamma), b_2(\gamma)]$ then $[A+B]^\gamma=[a_1(\gamma)+b_1(\gamma), a_2(\gamma)+b_2(\gamma)]$,  $[\lambda A]^\gamma=[\lambda a_1(\gamma), \lambda a_2(\gamma)]$ if $\lambda \geq 0$ and $[\lambda A]^\gamma=[\lambda a_2(\gamma), \lambda a_1(\gamma)]$ if $\lambda < 0$.

If $A_1, \ldots, A_n$ are fuzzy numbers and $\lambda_1, \ldots, \lambda_n \in {\bf R}$ then one can consider the fuzzy number $\displaystyle \sum_{i=1}^n \lambda_i A_i$.

A non--negative and monotone increasing function $f: [0, 1] \rightarrow {\bf{R}}$ is a {\em weighting function} if it
satisfies the normality condition $\int_0^1 f(\gamma)d\gamma=1$.

We fix a fuzzy number $A$ and a weighting function $f$ such that $[A]^\gamma=[a_1(\gamma), a_2(\gamma)]$ for all $\gamma \in [0, 1]$.

\begin{definitie}
\cite{fuller1}
The $f$--{\em weighted possibilistic expected value} of $A$ is defined by

$E(f, A)=\frac{1}{2} \int_0^1 (a_1(\gamma)+a_2(\gamma))f(\gamma)d\gamma$.
\end{definitie}

If $f(\gamma)=2\gamma$ for $\gamma \in [0, 1]$ then $E(f, A)$ is the possibilistic mean value introduced in \cite{carlsson1}.

\begin{propozitie}
If $A_1, \ldots, A_n$ are fuzzy numbers and $\lambda_1, \ldots, \lambda_n \in {\bf R}$ then $\displaystyle E(f, \sum_{i=1}^n \lambda_i A_i)=\sum_{i=1}^n \lambda_i E(f, A_i)$.
\end{propozitie}

\begin{remarca}
If $A$ is a fuzzy number then $E(f, A) \in [a_1(0), a_2(0)]=supp(A)$.
\end{remarca}

\begin{definitie}
\cite{appadoo}, \cite{zhang2} The $f$--{\em weighted possibilistic variance} of $A$ is defined by

$Var(f, A)=\frac{1}{2} \int_0^1 [(a_1(\gamma)-E(f, A))^2+(a_2(\gamma)-E(f, A))^2]f(\gamma)d\gamma$.
\end{definitie}

\begin{definitie}
\cite{appadoo}, \cite{zhang2} Let $A, B$ be two fuzzy numbers such that $[A]^\gamma=[a_1(\gamma), a_2(\gamma)]$ and
$[B]^\gamma=[b_1(\gamma), b_2(\gamma)]$ for any $\gamma \in [0, 1]$. The $f$--weighted covariance of $A$ and $B$
 is defined  by

$Cov(f, A, B)=\frac{1}{2} \int_0^1 [(a_1(\gamma)-E(f, A))(b_1(\gamma)-E(f, B))+(a_2(\gamma)-E(f, A))(b_2(\gamma)-E(f, B))]f(\gamma)d\gamma$.
\end{definitie}

For $f(\gamma)=2\gamma$, $\gamma \in [0, 1]$ the notions of possibilistic variance and possibilistic covariance appear in \cite{carlsson1}.

\begin{propozitie}
\cite{georgescu2} $Cov(f, A, B)=\frac{1}{2} \int_0^1 [a_1(\gamma) b_1(\gamma)+a_2(\gamma)b_2(\gamma)]f(\gamma)d\gamma-E(f, A)E(f, B)$
\end{propozitie}

\begin{corolar}
$Var(f, A)=\frac{1}{2} \int_0^1 [a_1^2(\gamma)+a_2^2(\gamma)]f(\gamma)d\gamma-E^2(f, A)$
\end{corolar}

\section{Multidimensional expected utilities}

In this section we recall the definitions and some properties of multidimensional possibilistic expected utility \cite{georgescu3} and mixed expected utility \cite{georgescu4}.

Let $u: {\bf{R}}^n \rightarrow {\bf{R}}$ be an $n$--dimensional utility function of class $C^2$. If $\vec{X}=(X_1, \ldots, X_n)$ is a random vector then $u(\vec{X})=u(X_1, \ldots, X_n)$ is a random variable and its expected value $M(u(\vec{X}))$ \footnote{When a random variable $X$ appears it can be inferred that it is reported to a probability space $(\Omega, \mathcal{K}, P)$ where $\Omega \subseteq {\bf {R}}$. The expected value of $X$ will be denoted $M(X)$.}  is called the (probabilistic) expected utility of $\vec{X}$ w.r.t. $u$.

The possibilistic correspondent of probabilistic expected utility is the possibilistic expected utility associated with a possibilistic vector, a weighting function and a utility function.

A {\emph possibilistic vector} has the form $\vec{A}=(A_1, \ldots, A_n)$ where each component $A_i$ is a fuzzy number.

We fix a weighting function $f$ and a utility function $u: {\bf{R}}^n \rightarrow {\bf{R}}$ \footnote{All utility functions in this paper will have the class $C^2$.}. Let $\vec{A}=(A_1, \ldots, A_n)$ be a possibilistic vector such that $[A_i]^\gamma=[a_i(\gamma), b_i(\gamma)]$ for any $\gamma \in [0, 1]$ and $i=1, \ldots, n$. We denote $\vec{a}(\gamma)=(a_1(\gamma), \ldots, a_n(\gamma))$ and $\vec{b}(\gamma)=(b_1(\gamma), \ldots, b_n(\gamma))$.

\begin{definitie}
\cite{georgescu3} The possibilistic expected utility of $\vec{A}$ w.r.t. $f$ and $u$ is defined by

(1) $E(f, u(\vec{A}))=\frac{1}{2} \int_0^1 [u(\vec{a}(\gamma)) + u(\vec{b}(\gamma))] f(\gamma) d\gamma$.
\end{definitie}

\begin{remarca}
If $n=1$ the definition of unidimensional expected utility from \cite{georgescu1} is found:

(2) $E(f, u(A_1))=\frac{1}{2} \int_0^1 [u(a_1(\gamma))+u(b_1(\gamma))]f(\gamma)d\gamma$.
\end{remarca}

\begin{remarca}
Let $n=1$.

(i) If $u : {\bf{R}} \rightarrow {\bf{R}}$ is the identity function then $E(f, u(A_1))=E(f, A_1)$.

(ii) If $u(x)=(x-E(f, A))^2$ for all $x \in {\bf R}$ then $E(f, u(A_1))=Var(f, A_1)$.

(iii) If $\lambda \in {\bf R}$ and $u(x)=\lambda$ for all $x \in {\bf R}$ then $E(f, u(A_1))=\lambda$.
\end{remarca}

\begin{remarca}
Let $n=2$ and $u(x, y)=(x-E(f, A_1))(y-E(f, A_2))$ for all $x, y \in {\bf R}$. Then $E(f, u(A_1, A_2))=Cov(f, A_1, A_2)$.
\end{remarca}

\begin{propozitie}
\cite{georgescu3} Let $g, h$ be two $n$--dimensional utility functions and $a, b \in {\bf R}$. If $u=ag+bh$ then $E(f, u(\vec{A}))=aE(f, g(\vec{A}))+bE(f, h(\vec{A}))$.
\end{propozitie}

A {\emph {mixed vector}} has the form $(\vec{A}, \vec{X})=(A_1, \ldots, A_n, X_1, \ldots, X_m)$ where $\vec{A}=(A_1, \ldots, A_n)$ is a possibilistic vector and $\vec{X}=(X_1, \ldots, X_m)$ is a random vector.

Let $u : {\bf{R}}^{n+m} \rightarrow {\bf{R}}$ be a utility function and $(\vec{A}, \vec{X})$ a mixed vector. Assume that $[A_i]^\gamma=[a_i(\gamma), b_i(\gamma)]$ for $\gamma \in [0, 1]$ and $i=1, \ldots, n$. For any $\gamma \in [0, 1]$ we denote

$u(\vec{a}(\gamma), \vec{X})(\omega)=u(a_1(\gamma), \ldots, a_n(\gamma), X_1(\omega), \ldots, X_n(\omega))$, $\omega \in \Omega$

$u(\vec{b}(\gamma), \vec{X})(\omega)=u(b_1(\gamma), \ldots, b_n(\gamma), X_1(\omega), \ldots, X_n(\omega))$, $\omega \in \Omega$

By these we define two functions $u(\vec{a}(\gamma), \vec{X}) : \Omega \rightarrow {\bf R}$ and $u(\vec{b}(\gamma), \vec{X}) : \Omega \rightarrow {\bf R}$ which obviously are random vectors.

\begin{definitie}
\cite{georgescu4} Let $(\vec{A}, \vec{X})$ be a mixed vector. The {\emph {mixed expected utility}} of $(\vec{A}, \vec{X})$ w.r.t. $f$ and $u$ is defined by

(3) $E(f, u(\vec{A}, \vec{X}))=\frac{1}{2} \int_0^1 [M(u(\vec{a}(\gamma), \vec{X}))+M(u(\vec{b}(\gamma), \vec{X}))]f(\gamma)d\gamma$.
\end{definitie}

The mixed expected utility generalizes both possibilistic expected utility and probabilistic expected utility.

\begin{propozitie}
\cite{georgescu4} Let $g, h$ be two $(m+n)$--dimensional utility
functions and $a, b \in {\bf{R}}$. If $u=ag+bh$ then $E(f,
u(\vec{A}, \vec{X}))=aE(f, g(\vec{A}, \vec{X}))+bE(f, h(\vec{A},
\vec{X}))$.
\end{propozitie}

\section{A possibilistic investment model}

In \cite{eeckhoudt2} p. 65 a probabilistic investment model is presented, in which during a fixed period an agent invests a wealth $w_0$ in a risk--free asset and in a risky asset. The risk--free asset is interpreted as a government bond and the risky asset as a stock. The return of the risky asset is a random variable. At the end of the actual period the agent invests the entire wealth $w_0$. The problem of the agent is to find that division of $w_0$ between bonds and stocks such that to realize a maximal gain.

In this section we will study a possibilistic version of the model of  \cite{eeckhoudt2} based on the hypothesis that the return of the risky asset is a fuzzy number.

Let $r$ be the risk--free return of bond and $x$ the value of stocks return. The agent invests the amount $\alpha$ in stocks and $w_0-\alpha$ in bonds. By \cite{eeckhoudt2} p. 66 the value of portfolio $(w_0-\alpha, \alpha)$ at the end of the period is

$(w_0-\alpha)(1+r)+\alpha (1+x)=w_0(1+r)+\alpha (x-r)=w+\alpha (x-r)$

where $w=w_0(1+r)$ is the future wealth obtained from bonds.

In the possibilistic model that we study $x$ is the value of a fuzzy number $A$. We consider the function

(1) $g(\alpha, w, x)=w+\alpha(x-r)$

If the fuzzy number $A$ is the return of stocks then the value of the portfolio $(w_0-\alpha, \alpha)$ at the end of the period is described by the fuzzy number $g(\alpha, w, A)=w+\alpha(A-r)$.

We fix a weighting function $f$. We assume that the agent has an increasing and concave utility function $u : {\bf {R}} \rightarrow {\bf {R}}$ of class $C^2$ and $[A]^\gamma=[a_1(\gamma), a_2(\gamma)]$, $\gamma \in [0, 1]$.

We consider the function:

(2) $h(\alpha, w, x)=u(g(\alpha, w, x))=u(w+\alpha(x-r))$

Then the mean gain associated with the portfolio $(w_0-\alpha, \alpha)$ will be the possibilistic expected utility associated with $f, A$ and $h$:

$K(\alpha, w)=E(f, h(\alpha, w, A))=\frac{1}{2} \int_0^1 [h(\alpha, w, a_1(\gamma))+h(\alpha, w, a_2(\gamma))]f(\gamma)d\gamma$

By (2)

(3) $K(\alpha, w)=\frac{1}{2} \int_0^1 [u(g(\alpha, w, a_1(\gamma)))+u(g(\alpha, w, a_2(\gamma)))]f(\gamma)d\gamma$

The investor's problem is to choose that $\alpha^\ast$ such that

(4) $K(\alpha^\ast, w)=\displaystyle \max_{\alpha} K(\alpha, w)$

\begin{propozitie}
(i) The function $K(\alpha, w)$ is concave in $\alpha$.

(ii) The real number $\alpha^\ast$ is the optimal solution of (4)
iff $\frac{\partial K(\alpha^\ast, w)}{\partial \alpha}=0$.
\end{propozitie}

\begin{proof}
(i) By (3) we obtain

$\frac{\partial K(\alpha, w)}{\partial \alpha}=\frac{1}{2} \int_0^1
[u'(g(\alpha, w, a_1(\gamma))) \frac{\partial g(\alpha, w,
a_1(\gamma))}{\partial \alpha}+u'(g(\alpha, w, a_2(\gamma)))
\frac{\partial g(\alpha, w, a_2(\gamma))}{\partial
\alpha}]f(\gamma)d\gamma$

But $\frac{\partial g(\alpha, w, x)}{\partial \alpha}=x-r$, therefore

(5) $\frac{\partial K(\alpha, w)}{\partial \alpha}=\frac{1}{2}
\int_0^1 [u'(g(\alpha, w, a_1(\gamma)))(a_1(\gamma)-r)+u'(g(\alpha,
w, a_2(\gamma)))(a_2(\gamma)-r)]f(\gamma)d\gamma$

From (5) it follows

(6) $\frac{\partial^2 K(\alpha, w)}{\partial \alpha^2}=\frac{1}{2}
\int_0^1 [u''(g(\alpha, w,
a_1(\gamma)))(a_1(\gamma)-r)^2+u''(g(\alpha, w,
a_2(\gamma)))(a_2(\gamma)-r)^2]f(\gamma)d\gamma$

The function $u$ is concave, therefore $u''(g(\alpha, w,
a_1(\gamma))) \leq 0$ and $u''(g(\alpha, w, a_2(\gamma)))\leq 0$,
therefore (6) implies that $\frac{\partial^2 K(\alpha, w)}{\partial
\alpha^2} \leq 0$ for any $\alpha$. Then $K(\alpha, w)$ is concave
in $\alpha$.

(ii) follows from (i).
\end{proof}

The following result is the possibilistic version of a theorem by Mossin \cite{mossin}.

\begin{propozitie}
Assume that $u'>0$ and $u''<0$.

(i) If $E(f, A)=r$  then $\alpha^\ast=0$.

(ii) If $r <E(f, A)$ then $\alpha^\ast >0$.
\end{propozitie}

\begin{proof}
(i) We notice that $g(0, w, x)=w$, therefore from (5) it follows

$\frac{\partial K(0, w)}{\partial \alpha}=\frac{1}{2} \int_0^1 [u'(w)(a_1(\gamma)-r)+u'(w)(a_2(\gamma)-r)]f(\gamma)d\gamma$

$=\frac{u'(w)}{2} \int_0^1 [a_1(\gamma)+a_2(\gamma)-2r]f(\gamma)d\gamma=u'(w)(E(f, A)-r)$

If $E(f, A)=r$ then $\frac{\partial K(0, w)}{\partial \alpha}=0$ then by Proposition 4.1 (ii), $\alpha^\ast=0$ is the optimal solution of (4).

(ii) Assume by absurdum that $\alpha^\ast \leq 0$. From $r < E(f, A)$ and $u'(w)>0$ it follows $\frac{\partial K(0, w)}{\partial \alpha}=u'(w)(E(f, A)-r)>0$. Since $u''<0$, $\frac{\partial K(\alpha, w)}{\partial \alpha}$ is decreasing in $\alpha$. Accordingly, $\alpha^\ast \leq 0$ implies $0=\frac{\partial K(\alpha^\ast, w)}{\partial \alpha} \geq \frac{\partial K(0, w)}{\partial \alpha} >0$. The contradiction shows that $\alpha^\ast >0$.
\end{proof}

The following result establishes an approximate calculation formula for the optimal solution $\alpha^\ast$ of problem (4).

\begin{propozitie}
$\alpha^\ast \approx - \frac{u'(w)}{u''(w)} \frac{E(f, A)-r}{Var(f, A)+(E(f, A)-r)^2}$
\end{propozitie}

\begin{proof}
We write the first--order Taylor approximation to $u'(w+\alpha(x-r))$ around $w$:

$u'(g(\alpha, w, x))=u'(w+\alpha(x-r))\approx u'(w)+\alpha (x-r) u''(w)$

Then, taking into account (5)

$\frac{\partial K(\alpha, w)}{\partial \alpha} \approx \frac{1}{2} \int_0^1 [(u'(w)+\alpha (a_1(\gamma)-r)u''(w))(a_1(\gamma)-r)]f(\gamma)d\gamma+$

$+ \frac{1}{2} \int_0^1 [(u'(w)+\alpha (a_2(\gamma)-r)u''(w))(a_2(\gamma)-r)]f(\gamma)d\gamma$

By denoting

$I_1=\frac{1}{2} \int_0^1 [a_1(\gamma)+a_2(\gamma)-2r]f(\gamma)d\gamma$

$I_2=\frac{1}{2} \int_0^1 [(a_1(\gamma)-r)^2 +(a_2(\gamma)-r)^2]f(\gamma)d\gamma$

it follows

$\frac{\partial K(\alpha, w)}{\partial \alpha} \approx  u'(w)I_1+ \alpha u''(w)I_2$.

We compute $I_1$ and $I_2$:

$I_1=\int_0^1 \frac{a_1(\gamma)+a_2(\gamma)}{2} f(\gamma)d\gamma -r=E(f, A)-r$

$I_2=\frac{1}{2} \int_0^1 [a_1^2(\gamma)+a_2^2(\gamma)-2r(a_1(\gamma)+a_2(\gamma))+2r^2]f(\gamma)d\gamma$

$=\frac{1}{2} \int_0^1 [a_1^2(\gamma)+a_2^2(\gamma)]f(\gamma)d\gamma-2r \int_0^1 \frac{a_1(\gamma)+a_2(\gamma)}{2} f(\gamma)d\gamma +r^2$

By Corollary 2.8, $\frac{1}{2} \int_0^1 [a_1^2(\gamma)+a_2^2(\gamma)]f(\gamma)d\gamma=Var(f, A)+E^2(f, A)$, thus

$I_2=Var(f, A)+E^2(f, A)-2rE(f, A)+r^2=Var(f, A)+(E(f, A)-r)^2$

Then

$\frac{\partial K(\alpha, w)}{\partial \alpha} \approx u'(w)(E(f, A)-r)+\alpha u''(w) [Var(f, A)+(E(f, A)-r)^2]$

Taking into account this relation, the approximate solution of the equation $\frac{\partial K(\alpha^\ast, w)}{\partial \alpha} =0$ has the form

$\alpha^\ast \approx - \frac{u'(w)}{u''(w)} \frac{E(f, A)-r}{Var(f, A)+(E(f, A)-r)^2}$.
\end{proof}

We recall from \cite{arrow0}, \cite{arrow}, \cite{pratt} the Arrow--Pratt index of the utility function $u$:

(7) $r_u(w)=-\frac{u''(w)}{u'(w)}$ for any $w \in {\bf{R}}$.

Then from Proposition 4.3 it follows

(8) $\alpha^\ast \approx \frac{1}{r_u(w)} \frac{E(f, A)-r}{Var(f, A)+(E(f, A)-r)^2}$

Consider two agents with the unidimensional utility functions $u_1$, $u_2$ such that $u_1'>0$, $u_2'>0$, $u_1''<0$, $u_2''<0$. We denote by $r_1(w)=r_{u_1}(w)$ and $r_2(w)=r_{u_2}(w)$ the Arrow-Pratt indexes of $u_1$ and $u_2$. We recall from \cite{eeckhoudt2}, p. 14 that $u_1$ is {\emph {more risk--averse}} than $u_2$ iff $r_1(w) \geq r_2(w)$ for any $w \in {\bf{R}}$.

\begin{corolar}
Let $\alpha_1^\ast, \alpha_2^\ast$ be the optimal solutions of problem (4) for the utility functions $u_1$ and $u_2$. If the agent $u_1$ is more risk--averse than $u_2$ then $\alpha_1^\ast \leq \alpha_2^\ast$.
\end{corolar}

\begin{proof}
The approximate formula (8) for $\alpha_1^\ast$ and $\alpha_2^\ast$ is applied.
\end{proof}

\begin{corolar}
If the Arrow--Pratt index $r_u(w)$ is decreasing then $\alpha^\ast(w)$ is increasing in wealth.
\end{corolar}

\section{Possibilistic background risk}

In the following we will study a two-asset model corresponding to a more complex situation than the one from the previous section. Among the initial data of the model, besides the wealth $w_0$, the existence of a possibilistic--type ''background risk'' is admitted. Our model will be a possibilistic version of the one of \cite{eeckhoudt2}, p. 68, in which the background risk is described by a random variable. In the interpretation of \cite{eeckhoudt2}, the background risk could be associated with the labor income.

As in Section 4, the agent invests the wealth $w_0$ in bonds and stocks. The return of bonds is the real number $r$ and the return of stocks is the fuzzy number $A$. A background risk represented by a fuzzy number $B$ is added. The agent invests the sum $\alpha$ in stocks and $w_0-\alpha$ in bonds.

We fix a weighting function $f$. Let $u: {\bf{R}} \rightarrow {\bf{R}}$ be the agent's utility function (of class $C^2$, increasing and concave). Assume that the level sets of the fuzzy numbers $A$ and $B$ have the form $[A]^\gamma=[a_1(\gamma), a_2(\gamma)]$, $[B]^\gamma=[b_1(\gamma), b_2(\gamma)]$ for any $\gamma \in [0, 1]$.

We consider the following functions:

(1) $g_1(\alpha, w, x, y)=w+y+\alpha(x-r)$

(2) $h_1(\alpha, w, x, y)=u(g_1(\alpha, w, x, y))=u(w+y+\alpha(x-r))$

The function $h_1(\alpha, w, ., .)$ will be considered a bidimensional utility function (with variables $x, y$), and $\alpha, w$ are parameters.

Then the fuzzy number $g_1(\alpha, w, A, B)$ will be the value of the portfolio $(w_0-\alpha, \alpha)$ at the end of the considered period. We also consider the possibilistic expected utility corresponding to the portfolio $(w_0-\alpha, \alpha)$:

(3) $K_1(\alpha, w)=E(f, h_1(\alpha, w, A, B))$

Then the investor's problem is to determine that $\alpha^{\ast \ast}$ for which

(4) $K_1(\alpha^{\ast \ast}, w)= \displaystyle \max_{\alpha} K_1(\alpha, w)$.

By (3), (2) and Definition 3.1 we have

$K_1(\alpha, w)=\frac{1}{2} \int_0^1 [h_1(\alpha, w, a_1(\gamma), b_1(\gamma))+h_1(\alpha, w, a_2(\gamma), b_2(\gamma))]f(\gamma)d\gamma$

$=\frac{1}{2} \int_0^1 [u(g_1(\alpha, w, a_1(\gamma), b_1(\gamma)))+u(g_1(\alpha, w, a_2(\gamma), b_2(\gamma)))]f(\gamma)d\gamma$

From this form of $K_1(\alpha, w)$ one obtains immediately

(5) $\frac{\partial K_1(\alpha, w)}{\partial \alpha}=\frac{1}{2} \int_0^1 u'(g_1(\alpha, w, a_1(\gamma), b_1(\gamma)))(a_1(\gamma)-r)f(\gamma)d\gamma+ $

$+\frac{1}{2} \int_0^1 u'(g_1(\alpha, w, a_2(\gamma), b_2(\gamma)))(a_2(\gamma)-r)f(\gamma)d\gamma$

(6) $\frac{\partial^2 K_1(\alpha, w)}{\partial \alpha^2}=\frac{1}{2} \int_0^1 u''(g_1(\alpha, w, a_1(\gamma), b_1(\gamma)))(a_1(\gamma)-r)^2 f(\gamma)d\gamma+ $

$+\frac{1}{2} \int_0^1 u''(g_1(\alpha, w, a_2(\gamma), b_2(\gamma)))(a_2(\gamma)-r)^2 f(\gamma)d\gamma$

By (5) and (6) we can prove:

\begin{propozitie}
(i) The function $K_1(\alpha, w)$ is concave in $\alpha$.

(ii) The real number $\alpha^{\ast \ast}$ is the solution of (4) iff $\frac{\partial K_1(\alpha^{\ast \ast}, w)}{\partial \alpha}=0$.
\end{propozitie}

\begin{propozitie}
$\alpha^{\ast \ast} \approx \alpha^\ast - \frac{Cov(f, A, B)+E(f, B)(E(f, A)-r)}{Var(f, A)+(E(f, A)-r)^2}$
\end{propozitie}

\begin{proof}
We consider the first--order approximation to $u'(w+y+\alpha(x-r))$ around $w$:

$u'(g_1(\alpha, w, x, y))=u'(w+y+\alpha(x-r))\approx u'(w)+u''(w)[y+\alpha(x-r)]f(\gamma)d \gamma$

Then, by (1) we can write

$\frac{\partial K_1(\alpha, w)}{\partial \alpha} \approx \frac{1}{2} \int_0^1 [u'(w)+(b_1(\gamma)+\alpha(a_1(\gamma)-r)u''(w)](a_1(\gamma)-r)f(\gamma)d\gamma$

$+\frac{1}{2} \int_0^1 [u'(w)+(b_2(\gamma)+\alpha(a_2(\gamma)-r)u''(w)](a_2(\gamma)-r)f(\gamma)d\gamma=$

$=u'(w)I_1+u''(w)I_2+\alpha u''(w)I_3$

where

$I_1=\frac{1}{2} \int_0^1 [a_1(\gamma)+a_2(\gamma)-2r]f(\gamma)d\gamma$

$I_2=\frac{1}{2} \int_0^1 [b_1(\gamma)(a_1(\gamma)-r)+b_2(\gamma)(a_2(\gamma)-r)]f(\gamma)d\gamma$

$I_3=\frac{1}{2} \int_0^1 [(a_1(\gamma)-r)^2 +(a_2(\gamma)-r)^2]f(\gamma)d\gamma$

By the proof of Proposition 4.3, $I_1=E(f, A)-r$ and $I_3=Var(f, A)+(E(f, A)-r)^2$. We compute $I_2$:

$I_2=\frac{1}{2} \int_0^1 [a_1(\gamma)b_1(\gamma)+a_2(\gamma)b_2(\gamma)]f(\gamma)d\gamma-r \int_0^1 \frac{b_1(\gamma)+b_2(\gamma)}{2} f(\gamma)d\gamma$.

By Proposition 2.7

$\frac{1}{2} \int_0^1 [a_1(\gamma)b_1(\gamma)+a_2(\gamma)b_2(\gamma)]f(\gamma)d\gamma=Cov(f, A, B)+E(f, A)E(f, B)$

thus

$I_2=Cov(f, A, B)+E(f, A)E(f, B)-rE(f, B)=Cov(f, A, B)+E(f, B)(E(f, A)-r)$

By replacing $I_1$, $I_2$ and $I_3$ in the approximate expression of  $\frac{\partial K_1(\alpha, w)}{\partial \alpha}$ we obtain

$\frac{\partial K_1(\alpha, w)}{\partial \alpha}\approx u'(w)(E(f, A)-r)+u''(w)[Cov(f, A, B)+E(f, B)(E(f, A)-r)]+\alpha u''(w)[Var(f, A)+(E(f, A)-r)^2]$

By this last relation, an approximate solution of the equation $\frac{\partial K_1(\alpha^{\ast \ast} , w)}{\partial \alpha}=0$ is

$\alpha^{\ast \ast}\approx -\frac{u'(w)}{u''(w)} \frac{E(f, A)-r}{Var(f, A)+(E(f, A)-r)^2}-\frac{Cov(f, A, B)+E(f, B)(E(f, A)-r)}{Var(f, A)+(E(f, A)-r)^2}$

By Proposition 4.3

$\alpha^{\ast \ast}\approx \alpha^{\ast} -\frac{Cov(f, A, B)+E(f, B)(E(f, A)-r)}{Var(f, A)+(E(f, A)-r)^2}$.
\end{proof}

\begin{corolar}
Let $u_1, u_2$ be unidimensional utility functions of two agents. Assume $u'_1>0$, $u'_2>0$, $u''_1<0$, $u''_2<0$. Let $\alpha_1^{\ast \ast}, \alpha_2^{\ast \ast}$ be the solutions of problem (4) for $u_1, u_2$. If the agent $u_1$ is more risk averse than $u_2$ then $\alpha_1^{\ast \ast} \leq \alpha_2^{\ast \ast}$.
\end{corolar}

\begin{proof}
Proposition 5.2 and Corollary 4.4 are applied.
\end{proof}

\begin{corolar}
If the Arrow--Pratt index of $u$ is decreasing then $\alpha^{\ast \ast}(w)$ is increasing in wealth.
\end{corolar}

According to intuition, adding a background risk to the risk generated by stocks should make the agent invest less in them, thus $\alpha^{\ast \ast}
\leq \alpha^{\ast}$. The following result gives a necessary and sufficient condition such that $\alpha^{\ast \ast}
\leq \alpha^{\ast}$.

\begin{propozitie}
$\alpha^{\ast \ast}\leq \alpha^{\ast}$ iff $r \leq \frac{Cov(f, A, B)+E(f, A)E(f, B)}{E(f, B)}$
\end{propozitie}

\begin{proof}
By noticing that $Var(f, A)+(E(f, A)-r)^2 \geq 0$ , by Proposition 5.2 the next equivalences follow:

$\alpha^{\ast \ast}
\leq \alpha^{\ast}$ iff $Cov(f, A, B)+E(f, B)(E(f, A)-r) \geq 0$

\hspace{1.5cm} iff  $r \leq \frac{Cov(f, A, B)+E(f, A)E(f, B)}{E(f, B)}$
\end{proof}

\section{A possibilistic model with probabilistic background risk}

In this section we will study a two-asset model characterized by the following initial data:

$\bullet$ the agent invests the wealth $w_0$ in bonds whose return is $r$ and in stocks whose return is the fuzzy number $A$

$\bullet$ there exists a background risk described by a random variable $Y$

We fix a weighting function $f$. Let $u: {\bf{R}} \rightarrow {\bf{R}}$ be the agent's utility function (of class $C^2$, increasing and concave). Assume that $[A]^\gamma=[a_1(\gamma), a_2(\gamma)]$ for $\gamma \in [0, 1]$. We consider the functions $g_1$ and $h_1$ defined in Section $5$:

$g_1(\alpha, w, x, y)=w+y+\alpha(x-r)$, $h_1(\alpha, w, x, y)=u(g_1(\alpha, w, x, y))$

In case of our model the objective function of the optimization problem will be the following mixed expected utility (corresponding to the mixed vector $(A, Y)$):

(1) $K_2(\alpha, w)=E(f, h_1(\alpha, w, A, Y))$

Then the optimization problem will be

(2) $K_2(\alpha^\Delta, w)=\displaystyle \max_{\alpha} K_2(\alpha, w)$

By Definition 3.6, $K_2(\alpha, w)$ is written

(3) $K_2(\alpha, w)=\frac{1}{2} [M(u(g_1(\alpha, w, a_1(\gamma), Y)))+M(u(g_1(\alpha, w, a_2(\gamma), Y)))]f(\gamma)d\gamma$

A common calculation shows that

(4) $\frac{\partial K_2(\alpha, w)}{\partial \alpha}=\frac{1}{2} \int_0^1 (a_1(\gamma)-r)M(u'(g_1(\alpha, w, a_1(\gamma), Y)))f(\gamma)d(\gamma)+$

$+\frac{1}{2} \int_0^1 (a_2(\gamma)-r)M(u'(g_1(\alpha, w, a_2(\gamma), Y)))f(\gamma)d(\gamma)$

(5)  $\frac{\partial^2 K_2(\alpha, w)}{\partial \alpha^2}=\frac{1}{2} \int_0^1 (a_1(\gamma)-r)^2 M(u''(g_1(\alpha, w, a_1(\gamma), Y)))f(\gamma)d(\gamma)+$

$+\frac{1}{2} \int_0^1 (a_2(\gamma)-r)^2 M(u''(g_1(\alpha, w, a_2(\gamma), Y)))f(\gamma)d(\gamma)$

With (4) and (5) one can prove

\begin{propozitie}
(i) The function $K_2(\alpha, w)$ is concave in $\alpha$.

(ii) The real number $\alpha^\Delta$ is the solution of problem (2) iff $\frac{\partial K_2(\alpha^\Delta, w)}{\partial \alpha}=0$.
\end{propozitie}

Let $\alpha^\ast$ be the solution of problem (4) of Section $4$ and $\alpha^\Delta$ the solution of problem (2) from above.

\begin{propozitie}
$\alpha^\Delta \approx \alpha^\ast - \frac{M(Y)(E(f, A)-r)}{Var(f, A)+(E(f, A)-r)^2}$
\end{propozitie}

\begin{proof}
Using the approximation $u'(g_1(\alpha, w, x, y)) \approx u'(w)+(y+\alpha(x-r))u''(w)$ from (4) one obtains

$\frac{\partial K_2(\alpha, w)}{\partial \alpha} \approx \frac{1}{2} \int_0^1 (a_1(\gamma)-r)f(\gamma)M(u'(w)+u''(w)(Y+\alpha(a_1(\gamma)-r)))d\gamma$

$+\frac{1}{2} \int_0^1 (a_2(\gamma)-r)f(\gamma)M(u'(w)+u''(w)(Y+\alpha(a_2(\gamma)-r)))d\gamma$

$=\frac{1}{2} \int_0^1 (a_1(\gamma)-r)f(\gamma)[u'(w)+u''(w)M(Y)+\alpha u''(w) (a_1(\gamma)-r)]d\gamma$

$+\frac{1}{2} \int_0^1 (a_2(\gamma)-r)f(\gamma)[u'(w)+u''(w)M(Y)+\alpha u''(w) (a_2(\gamma)-r)]d\gamma$

Denoting as in the proof of Proposition 4.3

$I_1=\frac{1}{2} \int_0^1 [a_1(\gamma)+a_2(\gamma)-2r]f(\gamma)d\gamma=E(f, A)-r$

$I_2=\frac{1}{2} \int_0^1 [(a_1(\gamma)-r)^2 +(a_2(\gamma)-r)^2]f(\gamma)d\gamma=Var(f, A)+(E(f, A)-r)^2$

then

$\frac{\partial K_2(\alpha, w)}{\partial \alpha} \approx (u'(w)+u''(w)M(Y))I_1+\alpha u''(w)I_2=$

$=[u'(w)+u''(w)M(Y)](E(f, A)-r)+\alpha u''(w)[Var(f, A)+(E(f, A)-r)^2]$

By this last relation the optimal solution of the equation $\frac{\partial K_2(\alpha^\Delta, w)}{\partial \alpha}=0$ is

$\alpha^\Delta \approx - \frac{u'(w)}{u''(w)}\frac{E(f, A)-r}{Var(f, A)+(E(f, A)-r)^2}- \frac{M(Y)(E(f, A)-r)}{Var(f, A)+(E(f, A)-r)^2}$

$=\alpha^\ast- \frac{M(Y)(E(f, A)-r)}{Var(f, A)+(E(f, A)-r)^2}$
\end{proof}

\begin{corolar}
Let $u_1, u_2$ be two unidimensional utility functions with $u_1'>0$, $u_2'>0$, $u_1''<0$, $u_2''<0$ and $\alpha_1^{\Delta}$, $\alpha_2^{\Delta}$ the solutions of problem (2) for $u_1$ and $u_2$. If the agent $u_1$ is more risk averse than $u_2$ then $\alpha_1^{\Delta} \leq \alpha_2^{\Delta}$.
\end{corolar}

\begin{proof}
By Proposition 6.2 and Corollary 4.4.
\end{proof}

\begin{corolar}
If the Arrow--Pratt index of the utility function $f$ is decreasing then $\alpha^\Delta(w)$ is increasing in wealth.
\end{corolar}

\begin{proof}
By Proposition 6.2 and Corollary 4.5.
\end{proof}

\begin{corolar}
$\alpha^\ast \leq \alpha^\Delta$ iff $M(Y)(E(f, A)-r) \leq 0$.
\end{corolar}

\section{A probabilistic model with possibilistic background risk}

The two-asset model discussed in this section has the following features:

$\bullet$ the agent invests the wealth $w_0$ in bonds with return $r$ and in stocks with return the random variable $X$

$\bullet$ the background risk is represented by a fuzzy number $B$

We fix a weighting function $f$. Let $u: {\bf{R}} \rightarrow {\bf{R}}$ be the agent's utility function (of class $C^2$, increasing and concave). Assume that $[B]^\gamma=[b_1(\gamma), b_2(\gamma)]$ for $\gamma \in [0, 1]$. We consider the functions

$g_1(\alpha, w, x, y)=w+y+\alpha(x-r)$, $h_1(\alpha, w, x, y)=u(g_1(\alpha, w, x, y))$.

$h_1(\alpha, w, ., .)$ is a bidimensional utility function and $(X, B)$ is a mixed vector. The objective function of the optimization problem of this model will be the following mixed expected utility:

(1) $K_3(\alpha, w)=E(f, h_1(\alpha, w, X, B))$

The the optimization problem associated with the portfolio $(w_0-\alpha, \alpha)$ is

(2) $K_3(\alpha^\bigtriangledown, w)=\displaystyle \max_{\alpha} K_3(\alpha, w)$

By Definition 3.6, $K_3(\alpha, w)$ has the following expression

$K_3(\alpha, w)=\frac{1}{2} \int_0^1 [M(h_1(\alpha, w, X, b_1(\gamma))+M(h_1(\alpha, w, X, b_2(\gamma))]f(\gamma)d\gamma$

$=\frac{1}{2} \int_0^1 [M(u(g_1(\alpha, w, X, b_1(\gamma))+M(u(g_1(\alpha, w, X, b_2(\gamma))]f(\gamma)d\gamma$

From this equality it follows immediately:

(3) $\frac{\partial K_3(\alpha, w)}{\partial \alpha}=\frac{1}{2} \int_0^1 M[u'(g_1(\alpha, w, X, b_1(\gamma)))(X-r)]f(\gamma)d\gamma+$

$+\frac{1}{2} \int_0^1 M[u'(g_1(\alpha, w, X, b_2(\gamma)))(X-r)]f(\gamma)d\gamma$

(4) $\frac{\partial^2 K_3(\alpha, w)}{\partial \alpha^2}=\frac{1}{2} \int_0^1 M[u''(g_1(\alpha, w, X, b_1(\gamma)))(X-r)^2]f(\gamma)d\gamma+$

$+\frac{1}{2} \int_0^1 M[u''(g_1(\alpha, w, X, b_2(\gamma)))(X-r)^2]f(\gamma)d\gamma$

Taking into consideration (3), (4) and that $u$ is concave one can easily prove

\begin{propozitie}
(i) The function $K_3(\alpha, w)$ is concave in $\alpha$.

(ii) The real number $\alpha^\bigtriangledown$ is the solution of (2) iff $\frac{\partial K_3(\alpha^\bigtriangledown, w)}{\partial \alpha}=0$.
\end{propozitie}

\begin{propozitie}
An approximate value of the solution $\alpha^\bigtriangledown$ of problem (2) is given by

(5) $\alpha^\bigtriangledown \approx -\frac{u'(w)}{u''(w)} \frac{M(X)-r}{M[(X-r)^2]}-\frac{E(f, B)(M(X)-r)}{M[(X-r)^2]}$.
\end{propozitie}

\begin{proof}
We will use here the approximation $u'(g_1(\alpha, w, x, y))\approx u'(w)+(y+(x-r)\alpha)u''(w)$. Then

$M[u'(g_1(\alpha, w, X, b_1(\gamma)))(X-r)]\approx$

$\approx M[u'(w)(X-r)+u''(w)b_1(\gamma)(X-r)+\alpha u''(w)(X-r)^2]=$

$=u'(w)(M(X)-r)+u''(w)b_1(\gamma)(M(X)-r)+\alpha u''(w)M[(X-r)^2]$

and analogously

$M[u'(g_1(\alpha, w, X, b_2(\gamma)))(X-r)]\approx$

$\approx u'(w)(M(X)-r)+u''(w)b_2(\gamma)(M(X)-r)+\alpha u''(w)M[(X-r)^2]$

By (4) it follows

$\frac{\partial K_3(\alpha, w)}{\partial \alpha}\approx u'(w)(M(X)-r)+u''(w)(M(X)-r)\int_0^1 \frac{b_1(\gamma)+b_2(\gamma)}{2}f(\gamma)d\gamma+$

$+\alpha u''(w)M[(X-r)^2]$

In conclusion

$\frac{\partial K_3(\alpha, w)}{\partial \alpha}\approx u'(w)(M(X)-r)+u''(w)E(f, B)(M(X)-r)\alpha u''(w)M[(X-r)^2]$

Then the equation $\frac{\partial K_3(\alpha^\bigtriangledown, w)}{\partial \alpha}=0$ has the approximate solution

$\alpha^\bigtriangledown \approx -\frac{u'(w)}{u''(w)} \frac{M(X)-r}{M[(X-r)^2]}-\frac{E(f, B)(M(X)-r)}{M[(X-r)^2]}$.
\end{proof}

\begin{remarca}
Using the Arrow--Pratt index $r_u(w)$, (5) can be written

(6) $\alpha^\bigtriangledown = \frac{1}{r_u(w)}\frac{M(X)-r}{M[(X-r)^2]}-\frac{E(f, B)(M(X)-r)}{M[(X-r)^2]}$.
\end{remarca}

\begin{corolar}
Let $u_1, u_2$ be two unidimensional utility functions with $u_1'>0$, $u_2'>0$, $u_1''<0$, $u_2''<0$ and $\alpha_1^{\bigtriangledown}$, $\alpha_2^{\bigtriangledown}$ the solutions of problem (2) for $u_1$ and $u_2$. If the agent $u_1$ is more risk averse than $u_2$ then $\alpha_1^{\bigtriangledown} \leq \alpha_2^{\bigtriangledown}$.
\end{corolar}

\section{Conclusions}
This paper has presented three investment models with two risk components: the investment risk (stocks) and background risk (labour income). Both the investment risk and the background risk can be probabilistic by random variables and possibilistic by fuzzy numbers.

The probabilistic approach from \cite{eeckhoudt2} assumes that both types of risk are random variables. The models of this paper cover the other three possibilities.

The first model is entirely possibilistic: both types of risk are fuzzy numbers. The optimization problem is formulated in terms of the notion of possibilistic expected utility of \cite{georgescu3}.

The other two models are mixed: a type of risk is probabilistic (random variable), the other is possibilistic (fuzzy number). The optimization problems are here formulated by the notion of mixed expected utility of \cite{georgescu4}.

For all three models approximate calculations formulas of the optimal solutions have been proved. These are expressed in terms of the Arrow--Pratt index of the utility function of the investor \cite{arrow}, \cite{mossin} and some possibilistic or probabilistic indicators (expected value, variance, covariance). These insure computation efficiency, but the error evaluation due to the approximation remains an open problem.

Finally, an application of these models to significant real situations remains the topic for possible future research. This would lead to a comparison of the four investment models, which might suggest to an investor the modality to use one of them.

\end{document}